\documentclass[11pt]{article}
\usepackage{hfstyle}
\usepackage{apacite}
\usepackage{amsmath}
 \usepackage{amsthm}
\usepackage{bbm}
\usepackage{graphicx}
\newcommand{\ZZ}{{\mathbbm{Z}}}
\newcommand{\NN}{{\mathbbm{N}}}

 \newtheorem{prop}{Proposition}
 \newtheorem{lemma}{Lemma}
\newcommand{\card}{\mathop{\mathrm{card}}}
\newcommand{\f}{{\mathbf f}}
\newcommand{\g}{{\mathbf g}}
\newcommand{\B}{{\mathcal B}}
\newcommand{\T}{{\mathcal T}}
\newcommand{\G}{{\mathcal N}}
\newcommand{\calS}{{\mathcal S}}

\begin{document}
\title{Response curves for cellular automata in one and two dimensions -- an example of rigorous
calculations}
\author{Henryk Fuk\'s and Andrew Skelton
      \oneaddress{
         Department of Mathematics\\
         Brock University\\
         St. Catharines, Ontario  L2S 3A1, Canada\\
         \email{hfuks@brocku.ca, andrewskelton123@msn.com}
       }
   }

\Abstract{
We consider the problem of computing a response curve for binary cellular automata -- that is, the curve describing the 
dependence of the density of ones after many iterations of the rule on the initial density of ones. We demonstrate
how this problem could be approached using rule 130 as an example. For this rule, preimage sets of finite strings exhibit
recognizable patterns, and it is therefore possible to compute both cardinalities of preimages of certain finite strings
and probabilities of occurrence of these strings in a configuration obtained by iterating a random initial configuration
$n$ times. Response curves can be rigorously calculated in both one- and two-dimensional versions of 
CA rule 130. We also discuss a special case of totally disordered initial configurations, that is, random configurations
where the density of ones and zeros are equal to 1/2.
}
\maketitle

\section*{Introduction}
Cellular automata (CA) can be viewed as computing devices, which take as an input
some initial configuration. The CA rule is iterated a number of times starting from this
configuration, resulting in a final configuration, which constitutes the output of the computation.
If the CA rule is complex, the above computation may be very difficult to characterize,
let alone understand in detail. If one considers ``simple'' CA rules, however, one can
say quite a lot about the process, as we shall see in what follows.

In many practical problems, e.g., in mathematical modelling, one wants to know
how a CA rule iterated over an initial configuration affects certain aggregate properties
of the configuration, such as, for example, the density of ones. If we take a randomly generated
initial configuration with a given density of ones, and iterate a given rule $n$ times 
over this configuration, what is the density of ones in the resulting configuration?
Using signal processing terminology, we want to know  the ``response curve'',  density of the
output as a function of the density of the input. Response curves appear in computational problems,
and a classical example of such a problem in CA theory is
the so-called density classification problem (DCP). If we denote the density of ones in the configuration
at time $n$ by $c_n$, the DCP asks us to find a rule for which $c_\infty=1$ if $c_0>1/2$ and
$c_\infty=0$ if $c_0<1/2$ -- that is, CA rule with  density response curve in a form of a step function.
Since it is known that such a rule does not exist \cite{LB95}, once could ask a related general question: which response
curves are possible in CA rules? Obviously, this problem is much more difficult than DCP, and not much
is known about it. We propose to approach this problem from an opposite direction: given the CA rule,
what can we say about its response curve? It turns out that in surprisingly many cases, the response
curve can be calculated exactly, providing that preimage sets of finite strings under the CA rule
exhibit recognizable patters. We will demonstrate this technique using as an example elementary CA rule 130,
in both one- and two-dimensional spaces.

\section*{Basic definitions}
Let $\G=\{0,1\}$ be called {\em a symbol set}, and let $\calS(\G)$
 be the set of all bisequences over $\G$, where by a bisequence we mean a
 function on  $\ZZ$ to $\G$. Throughout the remainder of this
text the configuration space  $\calS(\G)=\{0,1\}^{\ZZ}$ will be simply denoted by $\calS$.

{\em A block of length} $n$ is an ordered set $b_{0} b_{1}
\ldots b_{n-1}$, where $n\in \NN$, $b_i \in \G$.
Let $n\in \NN$ and let
$\B_n$ denote the set of all blocks of length $n$ over $\G$ and $\B$ be
the set of all finite blocks over $\G$.

For $r \in \NN$, a mapping $f:\{0,1\}^{2r+1}\mapsto\{0,1\}$ will be called {\em a cellular
 automaton rule of radius~$r$}. Alternatively, the function $f$ can be
 considered as a mapping of $\B_{2r+1}$ into $\B_0=\G=\{0,1\}$. 

Corresponding to $f$ (also called {\em a local mapping}) we define a
 {\em global mapping}  $F:\calS \to \calS$ such that
$
(F(s))_i=f(s_{i-r},\ldots,s_i,\ldots,s_{i+r})
$
 for any $s\in \calS$.
The {\em composition of two rules} $f,g$ can be now defined in terms
of
their corresponding global mappings $F$ and $G$ as $
(F\circ G)(s)=F(G(s)),
$
where $s \in \calS$. 

A {\em block evolution operator} corresponding to $f$ is a mapping
 $\f:\B \mapsto \B$ defined as follows. 
Let $r\in \NN$ be the radius of $f$, and let  $a=a_0a_1 \ldots a_{n-1}\in \B_{n}$
where $n \geq 2r+1 >0$. Then 
\begin{equation}
\f(a) = \{ f(a_i,a_{i+1},\ldots,a_{i+2r})\}_{i=0}^{n-2r-1}.
\end{equation}
 Note that if
$b \in B_{2r+1}$ then $f(b)=\f(b)$.

We will consider the case of $\G=\{0,1\}$ and $r=1$ rules,
 i.e., {\em elementary cellular automata}. In this case, when $b\in\B_3$,
then $f(b)=\f(b)$. The set 
$\B_3=\{000,001,010,011,100,101,110,111\}$ will be called the set of \textit{basic blocks}.

The number of $n$-step preimages of the block $b$ under the rule $f$
is defined as the number of elements of the set $\f^{-n}(b)$.
Given an elementary rule $f$, we will be especially interested in
the number of $n$-step preimages of basic blocks
under the rule $f$.

As mentioned in the introduction, we will use as an example rule 130 (using Wolfram's numbering scheme \cite{Wolfram94}) with local function defined as 
\begin{equation}
f\big(x, y, z\big) 
= \left\{ 
\begin{array}{l l}
  1 & \quad \text{if} \quad ( x \; y \; z) = ( 0 \; 0 \; 1) \text{ or } (1 \;1\;1),\\
 0 & \quad \text{otherwise.}\\
\end{array} \right.
\end{equation}
From now on, both the local function $f$ and the corresponding block evolution operator $\f$ will refer to rule 130, unless otherwise noted. 
We will calculate the response curve for this rule by considering the structure of preimages of finite strings, using the approach described
in earlier papers \cite{paper27,paper34,paper39}.
\section*{Structure of preimage sets in one dimension}
Since the only preimages of $1$ under rule 130 are blocks $111$ and $001$, let us consider first their preimages, that is, the
sets $\f^{-n}(111)$ and $\f^{-n}(001)$.  The following two propositions describe these sets.
\begin{prop}  \label{preim111-1d}
The set $\f^{-n}(111)$ has only one element, namely the block 
$\underbrace{11\ldots 1}_{2n+3}$, hence
$\card \f^{-n}(111)=1$.
\end{prop}
\begin{proof}
Exhaustively checking all 32 potential preimages of $111$ one can show that the only length-5 string~$b$ such that $\f(b)=111$ is $11111$. Since any block  of 1's consists of overlapping blocks \(111\), its preimage also must consist entirely of ones, and the proof by induction follows.
\end{proof}
\begin{prop} \label{preim001-1d}
The set $\f^{-n}(001)$  consists of all blocks of the form
\[ \underbrace{\star \hdots \star}_{2n-2i}\; 1\; 0\; 1 \;\underbrace{1 \hdots 1}_{2i} \quad \text{(if \(i\) is odd)} \quad\quad \text{or} \quad\quad \underbrace{\star \hdots \star}_{2n-2i}\; 0\; 0\; 1 \;\underbrace{1 \hdots 1}_{2i} \quad \text{(if \(i\) is even),}\]
where $i\in \{0\dots n\}$ and \(\star\) denotes an arbitrary value in \(\mathcal{S}\).
\end{prop}
\begin{proof}
By induction, if \(n=1\), then \(i \in \{0,1\}\), and our formula provides us with the following two  types of preimages,
$\star \star 001$ or $10111$. One can exhaustively check all 32 blocks of length five to verify that these are the only five blocks $b$ such that $\f(b)=001$.

Now, let us assume that we have the following \(n\)-step preimage:
\[\underbrace{\star \hdots \star}_{2n-2i}\; 0\; 0\; 1 \;\underbrace{1 \hdots 1}_{2i}\quad \quad\text{where \(i\) even and } i \in \{ 2\hdots n\} . \]
We find the preimages of this string, starting from the right and working toward the left. By considering all 16 blocks of length four, we can see that the only string that  has image  \(11\) under $\f$ is \(1111\), so we start with the following preimage (written above the string we are considering, with arrow indicating direction of proceeding):
\[\begin{tabular}{ccccccccccccccc}
& & & & &  &  &  &  & \(\leftarrow\) & 1& 1 & 1 & 1 \\
 \(\star\) & \(\star\) & \(\hdots\) & \(\star\) & \(\star\)& 0 & 0 & 1 &1 & \(\hdots\) & 1 & 1&1  \\
\end{tabular}\]
Continuing to the left, using the rule table of rule 130 we can construct the preimage up to the
following point:
\[\begin{tabular}{ccccccccccccccc}
& & &\(\leftarrow\) & 1& 0 & 1 & 1 & 1 &\(\hdots\) & 1& 1 & 1 & 1 \\
 \(\star\) & \(\star\) & \(\hdots\) & \(\star\) & \(\star\)& 0 & 0 & 1 &1 & \(\hdots\) & 1 & 1&1 \end{tabular}\]
From here to the left, all the remaining entries in the preimage are arbitrary:
 \[\begin{tabular}{ccccccccccccccc}
\(\star\) & \(\star\) & \(\hdots\)&\(\star\) & \(\star\) & 1& 0 & 1 & 1 & 1 &\(\hdots\) & 1& 1 & 1 & 1 \\
& \(\star\) & \(\star\) & \(\hdots\) & \(\star\) & \(\star\)& 0 & 0 & 1 &1 & \(\hdots\) & 1 & 1&1 \end{tabular}\]
We now have the following \((n+1)\)-step preimage:
 \[ \underbrace{\star \hdots\hdots\hdots \star}_{2(n+1)-2(i+1)}\; 1\; 0\; 1 \;\underbrace{1 \hdots 1}_{2(i+1)} \quad \quad \text{where \((i+1)\) is odd and }
 (i+1) \in \{ 3\hdots (n+1)\}. \]
  To finish the proof, we need to perform similar analysis for two other cases as follows.
  \begin{itemize}
  \item When $i=0$, a similar argument is used to find all $(n+1)$-step preimages in which $i=0,1$.
  \item When $i \in \{1,\dots n\}$ and odd, another similar argument can be used to find all $(n+1)$-step preimages in which $(i+1)$ is even.
  \end{itemize}
  We omit the details, but note that all possible $(n+1)$-step preimages are accounted for in the analysis, thus completing the induction step.
\end{proof}
We may now proceed with enumerating the elements of \(f^{-n}(001)\).
\begin{prop} For rule 130 in one dimension, we have
\( \card \f^{-n}( 0 0 1)  = \displaystyle \frac{4^{n+1}-1}{3}\).
\end{prop}
\begin{proof}
According to Proposition (\ref{preim001-1d}), for each value of \(i \in \{0\dots n\}\), there are \(2n-2i\) arbitrary values
in the preimage, so that there are $2^{2n-2i}$ of such preimages. Summing over $i$ we obtain 
\begin{equation}
 \card[f^{-n}( 0 0  1)] = \sum_{i=0}^{n} 2^{2n-2i} = 4^n \sum_{i=0}^{n} 4^{-i} = \frac{4^{n+1}-1}{3}.
\end{equation}
\end{proof}

Similar reasoning as in Proposition (\ref{preim111-1d}) leads to
\begin{equation}
\card \f^{-n}( 1 1 0) = \card\f^{-n}( 1 0 1) = 1. 
\end{equation}
It turns out that we now have enough information to find the cardinality of the preimage sets of each of the four remaining basic blocks, since not all of them are independent. The interdependence between preimage sets can be easily understood
if one considers this problem using the language of probability theory. Suppose that we start with an infinite
configuration in which all sites independently assume the value 0 or 1 with equal probability $1/2$. If we iterate
our rule $n$ times starting with this initial condition, one can show \cite{paper39} that the probability $P_n(b)$
of occurrence of block $b$ in the final configuration is given by
\begin{equation}
 P_n(b) = 2^{-|b|-2n} \card \f^{-n}(b),
\end{equation}
where $|b|$  is the length of the block $b$. Probabilities of different blocks are not independent,
since the following consistency conditions \cite{Dynkin69}
hold:
\begin{align} 
P_n(000) + P_n(001) &= P_n(100) + P_n(000) = P_n(00), \label{P000} \\
 P_n(110) + P_n(111)  &= P_n(011) + P_n(111)= P_n(11), \label{P110} \\
 P_n(010) + P_n(110)  &= P_n(100) + P_n(101)= P_n(10). \label{P010}
  \end{align}
Equation (\ref{P000}) implies that \(P_n(100) = P_n(001)\), hence
\begin{equation} \card \f^{-n}( 1 0 0) = \card \f^{-n}( 0 0  1). \end{equation}

Likewise, equations (\ref{P110}) and (\ref{P010}) imply that
\begin{equation} \card \f^{-n}( 0 1 1) = \card \f^{-n}( 1 1 0), \end{equation}
\begin{equation}  \card \f^{-n}( 0 1 0) = \card \f^{-n}( 1 0  1)+\card \f^{-n}( 1 0 0) -\card \f^{-n}( 1 1 0). \end{equation}

Finally, since we know the total number of all preimages of all basic blocks, we may subtract the 
seven known formulae from the total to find $\card \f^{-n}( 0 0  0)$. The results are 
summarized in Table~\ref{preim1dtable}.
\begin{table} 
\centering  
 \begin{tabular}{|| c | c || c | c ||}
 \hline
  b  &  $\card  \f^{-N}( b)$ & $b$ & $\card[ f^{-n}( b)]$\\
  \hline
 \( 0\;0\;0\) & \(  4^{n+1}-3   \) & \(1\;0\;0\) & \(   \frac{4^{n+1}-1}{3}  \)\\ \hline
    \(0\;0\;1\) & \(   \frac{4^{n+1}-1}{3}  \)  &  \(1\;0\;1\) & \(  1   \) \\ \hline
 \(0\;1\;0\) & \(  \frac{4^{n+1}-1}{3}   \) & \(1\;1\;0\) & \(    1   \)  \\ \hline
  \(0\;1\;1\) & \(  1   \) & \(1\;1\;1\) & \(   1  \) \\ 
 \hline \hline
 \end{tabular}
\caption{Number of Preimages of Basic Blocks of One-Dimensional Rule 130}\label{preim1dtable}
\end{table}

\section*{Dependence on the initial density in one dimension}
Let us now define $\rho$ to be the probability that a given cell is in state 1 in the randomly generated initial configuration. Obviously, $1-\rho$ is then the
probability that a cell is in state 0. Since in the initial configuration all cells are independent, the probability that a given block \(b\) will occur in the initial configuration is
\begin{equation}
P_0(b) = \rho^{\# \text{ of 1's in } b} (1-\rho)^{\# \text{ of 0's in } b}.
\end{equation}
Thus, if we wish to find the density of ones after $n$ iterations of rule \(f\), we can use preimages of $1$ and write
\begin{equation}
P_n(1) = \sum_{b \in f^{-n}(1)} P_0(b).
\end{equation}
Detailed derivation and discussion of the above equation can be found in \cite{paper39}, and will not
be repeated here.
Now, since the only preimages of $1$ are $001$ and $111$, one obtains
\begin{equation}
P_n(1) = P_{n-1} (111) + P_{n-1}(001)=
\sum_{b \in \f^{-n+1}(111)} P_0(b) + 
\sum_{b \in \f^{-n+1}(001)} P_0(b).
\end{equation}
We know the structure of the \(n\)-step preimage sets for Rule 130 in one dimension, thus we may compute desired probabilities in the above formula.
We know from Proposition \ref{preim111-1d}  that the preimage of \(111\) is comprised entirely of ones,
hence
\begin{equation}
P_n(111)=P_0(\underbrace{11\ldots 1}_{2n+3})=\rho^{2n+3}.
\end{equation}
From Proposition \ref{preim001-1d} we find,
\begin{align}
 P_n(001)&=\sum_{i \in \{0 \ldots n\}, i \,\,\text{is even}} (1-\rho)^2 \rho^{2i+1}+
\sum_{i \in \{0 \ldots n\}, i \,\,\text{is odd}} (1-\rho) \rho^{2i+2}\\
&=\sum_{k=0}^{\lceil \frac{n-1}{2}\rceil} (1-\rho)^2 \rho^{4k+1}+
\sum_{k=0}^{\lfloor \frac{n-1}{2}\rfloor} (1-\rho) \rho^{4k+4}\\
&=\frac {\rho\, \left( -  {\rho} ^{4 \lceil 
(n-1)/2 \rceil +4 }+  {\rho} ^{4
 \lceil (n-1)/2 \rceil +5} -  {\rho}^{ 4 \lfloor (n+1)/2 \rfloor +3}+{\rho}^{3
} - \rho + 1\right) }{{\rho}^{3}+{\rho}^{2}+\rho+1}.
\end{align}
Using the fact that $P_n(1) = P_{n-1} (111) + P_{n-1}(001)$, we obtain the equation of the \textit{response curve}, that
is, dependence of the density of ones after $n$ steps, denoted by $c_n=P_n(1)$, on  the initial density, denoted by $\rho$:
\begin{equation}
c_n=\rho^{2n+1}+
\frac {\rho\, \left( -  {\rho} ^{4 \lceil 
(n-2)/2 \rceil +4 }+  {\rho} ^{4
 \lceil (n-2)/2 \rceil +5} -  {\rho}^{ 4 \lfloor n/2 \rfloor +3}+{\rho}^{3
} - \rho + 1\right) }{{\rho}^{3}+{\rho}^{2}+\rho+1}.
\end{equation}
By taking the limit of $n \to \infty$, we obtain the asymptotic response curve,
\begin{equation} \label{densityesponse1d}
  c_\infty=
 \left\{ \begin{array}{ll}
 \displaystyle \frac{ \rho(\rho^3 - \rho + 1)}{\rho^3 + \rho^2 + \rho + 1}  & \mbox{if $\rho<1$},  \vspace{10pt} \\
 1    & \mbox{if $\rho=1$}.
\end{array}
\right.
\end{equation}
We performed computer simulations to illustrate the formula (\ref{densityesponse1d}) for the the dependence of $c_\infty$ on the initial density $\rho$. We considered an initial configuration of \(5000\) cells, and varied the initial density from 0 to 100\(\%\), increasing it by a step size of 1\(\%\), iterating rule 130 until we reached a fixed density. Results were then averaged over \(20\) runs for each initial density. The results are presented in Figure \ref{respurvefig1d}. As we can see,
the response curve calculated above for infinite configurations agrees very well with simulations performed
on a finite lattice. 
\begin{figure}
\begin{center}
\includegraphics[scale=0.4]{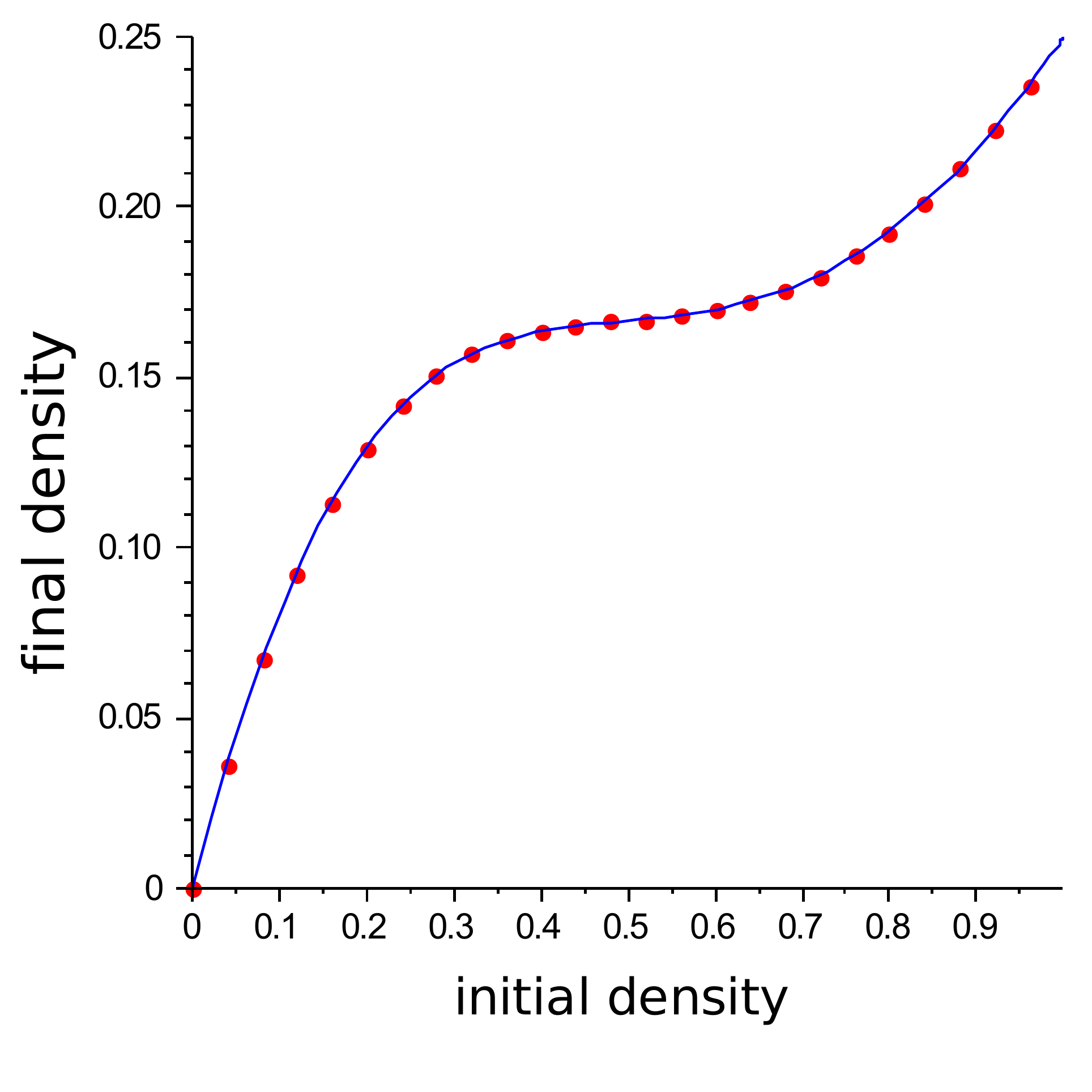}	
\end{center}
\caption{Plots of the response curve, i.e., dependence of $c_\infty$ on $\rho$ for one-dimensional rule 130. Circles  correspond to computer simulations using a lattice of 5000 sites, as described in the text, while
the continuous line represents theoretical curve.}\label{respurvefig1d}
\end{figure}
As a final remark in this section, let us note that for the special case of $\rho=1/2$ we obtain
\begin{equation}\label{dens1d}
 c_n = \frac{1}{6} + \frac{1}{3} 4^{-n}. \end{equation}
We can see that the convergence toward $c_\infty=1/6$ is exponential.

We will now proceed to consider two-dimensional version of rule 130. As it turns out, the basic ideas used in the preceding section for calculating response curve 
in one dimension can be carried over to two dimensions.
\section*{Basic definitions in two dimensions}
All basic definitions introduced at the beginning of this paper  can be easily generalized to two dimensions. We will consider a variant of rule 130 in
which the state of the cell at the next iteration depends on the cell itself, its right neighbour and the top neighbour. We will call this type of neighbourhood
an $L$-shaped neighbourhood. 
 To be more precise, 
let us define the local mapping of a \textit{two-dimensional cellular automaton with $L$-shaped neighbourhood} as 
\begin{equation}
g: (\begin{smallmatrix} x \\ y & z \end{smallmatrix}) \to \{0,1\},
\end{equation}
 where \(x,y,z  \in G\).
 Similarly as in one dimension, $g$ has a corresponding global mapping, $G: \{0,1\}^{\ZZ ^2} \to \{0,1\}^{\ZZ ^2}$
such that 
\begin{equation}
\big(G(s)\big)_{i,j} = g \big(\begin{smallmatrix} s_{i,j+1} \\ s_{i,j} & s_{i+1,j} \end{smallmatrix}\big) 
\end{equation}
for any
$s \in \{0,1\}^{\ZZ ^2}$. Blocks in two dimensions will be defined as regions of 2D lattice  in the shape of isosceles right  triangles. 
We define the set of triangular blocks of size $r$  as the set consisting of elements
\begin{equation}\label{blockb}
\begin{tabular}{c @{\hspace{1mm}} c @{\hspace{1mm}} c} 
\\ $b_ {1,r}$\\ \(\vdots\) & \(\ddots\) \\ $b_{1,1}$ & \(\hdots\) & $b_{r,1}$,
 \end{tabular} 
\end{equation}
where each $b_{i,j}\in \{0,1\}$. This set will be denoted by $\T_r$. The block evolution operator $\g: \T_r \to \T_{r-1}$ will be defined as a function
which transforms triangular block (\ref{blockb}) into another block
\begin{equation}\label{blockc}
\begin{tabular}{c @{\hspace{1mm}} c @{\hspace{1mm}} c} 
\\ $c_ {1,r-1}$\\ \(\vdots\) & \(\ddots\) \\ $c_{1,1}$ & \(\hdots\) & $c_{r-1,1}$,
 \end{tabular} 
\end{equation}
where $c_{i,j}=g \big(\begin{smallmatrix} b_{i,j+1} \\ b_{i,j} & b_{i+1,j} \end{smallmatrix}\big)$ for 
$i \in \{1,\ldots, r-1 \}$, $j \in \{1, \ldots, r-i \}$. 

Finally, rule 130 is defined in two dimensions as
\begin{equation}
g \big(\begin{smallmatrix} x &\\ y & z \\ \end{smallmatrix}\big)
= \left\{ 
\begin{array}{l l}
  1 & \quad \text{if} \quad \big(\begin{smallmatrix} x &\\ y & z \\ \end{smallmatrix}\big) = \big(\begin{smallmatrix} 0 &\\ 0 & 1 \\ \end{smallmatrix}\big) \text{ or } \big(\begin{smallmatrix}1&\\ 1 & 1 \\ \end{smallmatrix}\big),\\
 0 & \quad \text{otherwise.}\\
\end{array} \right.
\end{equation}
Since the only preimages of $1$ are triangular blocks $\left(\begin{smallmatrix} 0 &\\0 & 1 \\ \end{smallmatrix}\right)$ and
$\left(\begin{smallmatrix} 1 &\\1 & 1 \\ \end{smallmatrix}\right)$, we need, similarly as in one dimension, to analyze structure of preimage sets
$\g^{-n}(\begin{smallmatrix} 0 &\\0 & 1 \\ \end{smallmatrix})$ and $\g^{-n}(\begin{smallmatrix} 1 &\\1 & 1 \\ \end{smallmatrix})$.
\section*{Structure of preimage sets in two dimensions}
\begin{prop}\label{preim001-2d}
 The set $\g^{-n}(\begin{smallmatrix} 0 &\\0 & 1 \\ \end{smallmatrix})$ consists of all blocks of the following form:
\begin{equation} \label{preim001-2d-block}
 \begin{tabular}{c @{\hspace{2mm}} c @{\hspace{2mm}} c @{\hspace{+2mm}} c @{\hspace{+2mm}} c @{\hspace{2mm}} c @{\hspace{2mm}} c @{\hspace{2mm}} c} 
\\ $\star$
\vspace{-0.5mm} \\ $\vdots$ & $\ddots$ 
\vspace{-0.5mm} \\ $\vdots$ & & $\star$
\vspace{-0.5mm} \\ $\vdots$ & & $\vdots$ & $a_{1}$ 
\vspace{-0.5mm} \\ $\vdots$ & & $\vdots$ & $a_{2}$ & 1 
\vspace{-0.5mm} \\ $\vdots$ & & $\vdots$ & $a_{3}$ & $\vdots$ & $\ddots$
\vspace{-0.5mm}\\ $\vdots$ & & $\vdots$ & $\vdots$ & $\vdots$ & &$\ddots$ 
\vspace{-0.5mm} \\ $\star$ & $\hdots$ & $\star$ & $a_{i+2}$ & 1 & $\hdots$ & $\hdots$ & 1
\\ \multicolumn{3}{c}{$\underbrace{\;\;\;\;\;\;\;\;\;\;\;\;\;}_\text{n-i}$}  & & \multicolumn{4}{c}{$\underbrace{\;\;\;\;\;\;\;\;\;\;\;\;\;\;\;\;\;}_\text{i+1}$} 
 \end{tabular}  
\end{equation}
where $i \in \{0 \dots n\}$ and for $a_3 \ldots a_{i+2}$ we take arbitrary values. If $i=0$, then $(a_1,a_2)=(0,0)$, while if $i>0$, $a_1,a_2$ are determined
by
\begin{eqnarray}
a_1&=&1+ \sum_{j=2}^{i+1}{\binom{i}{j-1}}a_j \,\,\,\mathrm{mod} \,2, \label{a1}\\
a_2&=&1+ \sum_{j=2}^{i+1}{\binom{i}{j-1}}a_{j+1} \,\,\,\mathrm{mod} \,2, \label{a2}
\end{eqnarray}
\end{prop}
Before we attempt the proof, we need to observe two facts. First of all, note that
\begin{equation}
 g(\begin{smallmatrix} x &\\y & 1 \\ \end{smallmatrix})=x+y+1 \mod 2,
\end{equation}
and let us define $h(x,y)=x+y+1 \mod 2$. Secondly, consider the following procedure.
We start with a binary sequence $a_1,a_2,...a_k$ and replace it by 
a sequence of pairs $h(a_1,a_2), h(a_2,a_3), \ldots, h(a_{k-1},a_k)$. This 
new sequence is obviously of length $k-1$. If we repeat this process $k-1$ times,
we will end up with just one number, to be denoted by $q(a_1,a_2,\ldots,a_k)$. The following
lemma can be easily proved by induction.
\begin{lemma} \label{lemmaq}
$\displaystyle q(a_1,a_2,\ldots,a_k)= 2^{k-1} -1 + \sum_{i=1}^{k} {\binom{k-1}{i-1}} a_i \mod 2$.
\end{lemma}
\begin{proof}
When $k=2$, one obviously has $q(a_1,a_2)=h(a_1,a_2)=a_1+a_2+ 1 \mod 2$, 
and 
\begin{equation}
 2^{2-1} - 1 + \sum_{i=1}^{2} {\binom{1}{i-1}} a_i \mod 2 =1+ a_1+a_2 \mod 2,
\end{equation}
thus the lemma is indeed true for $k=2$.

 Suppose now that it holds for a given $k$, and let us consider a 
binary sequence  $a_1,a_2,\ldots, a_{k+1}$. We clearly have
\begin{equation}
 q(a_1,a_2,\ldots,a_{k+1})=q(a_1,a_2,\ldots,a_k) + q(a_2,a_2,\ldots,a_{k+1}) + 1 \mod 2,
\end{equation}
hence
\begin{align}
 q(a_1,a_2,\ldots,a_{k+1})& = 2^{k-1}-1 + \sum_{i=1}^{k} {\binom{k-1}{i-1}} a_i 
+ 2^{k-1}-1 + \sum_{i=2}^{k+1} {\binom{k-1}{i-2}} a_i  + 1 \mod 2 \\
&=2^k-1 + a_1 + \sum_{i=2}^{k} {\binom{k-1}{i-1}} a_i + a_{k+1}
+ \sum_{i=2}^{k} {\binom{k-1}{i-2}} a_i \mod 2.
\end{align}
Using Pascal's identity, this becomes
\begin{equation}
 q(a_1,a_2,\ldots,a_{k+1}) = 2^k-1 + a_1 + a_k   + \sum_{i=2}^{k} {\binom{k}{i-1}}    \mod 2.
\end{equation}
Using $\binom{k}{0} = \binom{k}{k}=1$ we can incorporate $a_1$ and $a_{k+1}$ into the sum, obtaining
\begin{equation}
 q(a_1,a_2,\ldots,a_{k+1}) = 2^k-1 + \sum_{i=1}^{k+1} {\binom{k}{i-1}}    \mod 2,
\end{equation}
which is the desired formula for $k+1$. Proof by induction is therefore complete.

\end{proof}

Having this lemma, we can now proceed with the sketch of the proof of the Proposition \ref{preim001-2d}.
\begin{proof}
Consider first the case of $n=1$, when the only preimages we obtain are of the form
\begin{equation}\label{preimages001}
\begin{matrix} \star \\ \star & 0 \\ \star & 0 & 1 \end{matrix} \hspace{10mm} \text{or} \hspace{10mm}  \begin{matrix} a_1 \\ a_2 & 1\\ a_3 & 1 & 1 \end{matrix}
\end{equation}
where $a_1=a_3$ and $a_2=1+a_3\mod 2$ are given by eq. (\ref{a1}) and (\ref{a2}) for any $a_3 \in \{0,1\}$. One can check that these are indeed the only desired  preimages by
applying $\g$ to all blocks $\T_3$ and verifying that only the blocks of the above form produce $\begin{smallmatrix} 0 &\\0 & 1 \\ \end{smallmatrix}$.
This means that the block $\begin{smallmatrix} 0 &\\0 & 1 \\ \end{smallmatrix}$ can appear in a configuration by
two ways, by moving from the left one unit at a time, or by being created from a configuration shown on the right hand side of (\ref{preimages001}).

Consider a number of steps larger than one, denoted by $n$. A given block may appear in a certain place because it was created in another location, $k$ units to the right, and then moved to the desired place from there in $k$ steps, where $k \leq n$, or because it appeared in this place as a result of the second configuration
of eq. (\ref{preimages001}). 

Now note that the block shown in eq. (\ref{preim001-2d-block}) 
in $i$ iterations will produce block
\begin{equation} \label{blockresult}
 \begin{tabular}{c @{\hspace{2mm}} c @{\hspace{2mm}} c @{\hspace{+2mm}} c @{\hspace{+2mm}} c @{\hspace{2mm}} c @{\hspace{2mm}} c @{\hspace{2mm}} c} 
\\ $?$
\vspace{-0.5mm} \\ $\vdots$ & $\ddots$ 
\vspace{-0.5mm} \\ $\vdots$ & & $?$
\vspace{-0.5mm} \\ $\vdots$ & & $\vdots$ & 0
\vspace{-0.5mm} \\ $?$ & $\hdots$ & $?$ & $ 0$ & 1,
\\ \multicolumn{3}{c}{$\underbrace{\;\;\;\;\;\;\;\;\;\;\;\;\;}_\text{n-i}$}  & 
 \end{tabular}
\end{equation}
where $?$ denotes some value resulting from iteration of the rule.
This is because the triangular block of ones in eq. (\ref{preim001-2d-block}) will shrink until
it becomes a single 1, as shown above, and the column $a_1,a_2,\ldots a_m$ will also shrink until 
it becomes a pair $(u,v)$, where $u=q(a_1,a_2,\ldots, a_{i+1})$ and 
$v=q(a_2,a_3,\ldots,a_{i+2})$.  The reason for this is the fact that  iteration of the rule is
equivalent to iterative application of  $h(x,y)$ to $a_1,a_2,\ldots a_{i+2}$.
Conditions (\ref{a1}) and (\ref{a2})  are equivalent to 
\begin{eqnarray}
1&=& \sum_{j=1}^{i+1}{\binom{i}{j-1}}a_j \,\,\,\mathrm{mod} \,2, \\
1&=& \sum_{j=1}^{i+1}{\binom{i}{j-1}}a_{j+1} \,\,\,\mathrm{mod} \,2,
\end{eqnarray}
and by Lemma \ref{lemmaq} this implies $u=0$, $v=0$, that is, we obtain column of two zeros, as shown in
eq. (\ref{blockresult}).

Iterating block shown in eq. (\ref{blockresult}) further, after $n-i$ iterations,  $\begin{smallmatrix} 0 &\\0 & 1 \\ \end{smallmatrix}$ will move to the left one step at a time.
In the end,  as a result of applying
operator $\g$ $n$-times to block (\ref{preim001-2d-block}),  $\begin{smallmatrix} 0 &\\0 & 1 \\ \end{smallmatrix}$ will be produced. Since $i$ varies from $0$ to $n$ this means that indeed all possibilities
of arriving from some place on the right ($i=1, \ldots n-1$) as well as being ``created in place'' ($i=n$) are covered.
\end{proof}
Using very similar argument, one can prove analogous proposition for block $\begin{smallmatrix} 1 &\\1 & 1 \\ \end{smallmatrix}$.
\begin{prop} \label{preim111-2d}
The set, \(\g^{-n}(\begin{smallmatrix} 1 &\\1 & 1 \\ \end{smallmatrix})\), consists of all blocks in the set \(\bigcup_{i=0}^{n} (A_i \setminus B_i) \cup C\), where, for a fixed value of $i$, 
$A_i$ and $B_i$ are, respectively, the sets of all blocks of the form
\[ \begin{tabular}{c @{\hspace{2mm}} c @{\hspace{2mm}} c @{\hspace{+1mm}} c @{\hspace{+1mm}} c @{\hspace{2mm}} c @{\hspace{2mm}} c @{\hspace{2mm}} c} 
\\ \(\star\)
\vspace{-0.5mm} \\ \(\vdots\) & \(\ddots\) 
\vspace{-0.5mm} \\ \(\vdots\) & & \(\star\)
\vspace{-0.5mm} \\ \(\vdots\) & & \(\vdots\) & \(b_1\) 
\vspace{-0.5mm} \\ \(\vdots\) & & \(\vdots\) & \(b_2\) & 1 
\vspace{-0.5mm} \\ \(\vdots\) & & \(\vdots\) & \(b_3\) & \(\vdots\) & \(\ddots\)
\vspace{-0.5mm}\\ \(\vdots\) & & \(\vdots\) & \(\vdots\) & \(\vdots\) & &\(\ddots\) 
\vspace{-0.5mm} \\ \(\star\) & \(\hdots\) & \(\star\) & \(b_{i+2}\) & 1 & \(\hdots\) & \(\hdots\) & 1,
\\ \multicolumn{3}{c}{\(\underbrace{\;\;\;\;\;\;\;\;\;\;\;\;\;}_\text{n-i}\)}  & & \multicolumn{4}{c}{\(\underbrace{\;\;\;\;\;\;\;\;\;\;\;\;\;\;\;\;\;}_\text{i+1}\)} 
 \end{tabular}  \hspace{20mm}
 \begin{tabular}{c @{\hspace{2mm}} c @{\hspace{2mm}} c @{\hspace{+3mm}} c @{\hspace{+3mm}} c @{\hspace{2mm}} c @{\hspace{2mm}} c @{\hspace{2mm}} c} 
\\ & \\&\\\(\star\) \\ \(\vdots\) & \(\ddots\) \\ \(\vdots\) & & \(\star\) \\ \(\vdots\) & & \(\vdots\) & 1
\\ \(\vdots\) & & \(\vdots\) & \(\vdots\) & \(\ddots\) 
\\ \(\star\) & \(\hdots\) & \(\star\) & 1 & \(\hdots\) & 1,
\\ \multicolumn{3}{c}{\(\underbrace{\;\;\;\;\;\;\;\;\;\;\;}_\text{n-i}\)} &\multicolumn{3}{c}{\(\underbrace{\;\;\;\;\;\;\;\;\;\;}_\text{i+2}\)} 
 \end{tabular} \]
and $C$ is the set whose only element is the block
\[
 \begin{tabular}{c @{\hspace{2mm}} c @{\hspace{2mm}} c} 
\\ 1 \\ \(\vdots\) & \(\ddots\) \\ 1 & \(\hdots\) & 1,
\\ \multicolumn{3}{c}{\(\underbrace{\;\;\;\;\;\;\;\;}_\text{n+2}\)}
 \end{tabular} 
 \]
 where \(i \in \{0 \dots n\}\) and for \(b_3 \ldots b_{i+2}\) we take arbitrary values. If $i=0$, then $(b_1,b_2)=(1,1)$ and if $i>0$, then \(b_1\), \(b_2\) are determined by
 \begin{equation} \label{b1}
 b_1 = \sum_{j=2}^{i+1} \binom{i}{j-1} b_j \mod 2,
 \end{equation}
 \begin{equation} \label{b2}
 b_2 = \sum_{j=2}^{i+1} \binom{i}{j-1} b_{j+1} \mod 2.
 \end{equation}
\end{prop}

\begin{prop} The number of \(n\)-step preimages of blocks \(\begin{smallmatrix} 0 &\\0 & 1 \\ \end{smallmatrix}\) 
and  \(\begin{smallmatrix} 1 &\\1 & 1 \\ \end{smallmatrix}\) 
is given by
\begin{align}
 \card [\g^{-n} \big(\begin{smallmatrix} 0 &\\0 & 1 \\ \end{smallmatrix} \big)]  &= 2^{\frac{n^2+5n}{2}} \sum_{i=0}^{n} 2^{-\frac{i(i+3)}{2}}  \\
 \text{card}[\g^{-n}\big(\begin{smallmatrix} 1 &\\1 & 1 \\ \end{smallmatrix} \big) ] &= 2^{\frac{n^2+5n}{2}} \Big( 4-3 \sum_{i=0}^{n} 2^{-\frac{i(i+3)}{2}} \Big),
\end{align}

\end{prop}

\begin{proof} For a given value of \(i\), we have \((n+2)+(n+1)+\dots+(i+3) + i= \frac{n(n+5)-i(i+5)}{2} +i\) 
arbitrary values represented by \(\star\) in the preimage of  \(\begin{smallmatrix} 0 &\\0 & 1 \\ \end{smallmatrix}\). Therefore, for each \(i\), there are \(2^{\frac{n(n+5)-i(i+5)}{2}+i} = 2^{\frac{n(n+5)-i(i+3)}{2}}\) possible configurations. Summing over all values of \(i\in \{0 \dots n \}\) gives the desired result.

In preimages of  \(\begin{smallmatrix} 1 &\\1 & 1 \\ \end{smallmatrix}\), for a given $i$ we have
\(2^{\frac{n(n+5)-i(i+3)}{2}}\) possible blocks in set $A_i$, and \(2^{\frac{n(n+5)-i(i+5)}{2}}\) blocks in set $B_i$.
Summing over $i$ and adding one block consisting of all ones, we obtain
\begin{equation}
 \text{card}[\g^{-n}\big(\begin{smallmatrix} 1 &\\1 & 1 \\ \end{smallmatrix} \big) ]
=\sum_{i=0}^n 2^{\frac{n(n+5)-i(i+3)}{2}} - \sum_{i=0}^n 2^{\frac{n(n+5)-i(i+5)}{2}} +1.
\end{equation}
After changing the summation index $i$  in the second sum to $j=i+1$, both sums can be combined together,
and after simplification one obtains
\begin{equation}
\text{card}[\g^{-n}\big(\begin{smallmatrix} 1 &\\1 & 1 \\ \end{smallmatrix} \big) ] = 2^{\frac{n^2+5n}{2}} \Big( 4-3 \sum_{i=0}^{n} 2^{-\frac{i(i+3)}{2}} \Big).
\end{equation}
\end{proof}
Similar analysis can be performed for block  \(\begin{smallmatrix} 1 &\\1 & 0 \\ \end{smallmatrix}\), leading
to the same number of preimages as in the case of  \(\begin{smallmatrix} 1 &\\1 & 1 \\ \end{smallmatrix}\),
\begin{equation}
\text{card}[\g^{-n}\big(\begin{smallmatrix} 1 &\\1 & 0 \\ \end{smallmatrix} \big) ] = 2^{\frac{n^2+5n}{2}} \Big( 4-3 \sum_{i=0}^{n} 2^{-\frac{i(i+3)}{2}} \Big).
\end{equation}
There are two other blocks which have easy to describe preimage sets, as demonstrated in the next two propositions.
\begin{prop}\label{preim011-2d}
 The set $\g^{-n}(\begin{smallmatrix} 0 &\\1 & 1 \\ \end{smallmatrix})$ consists of all blocks of the following form:
\begin{equation} \label{preim011-2d-block}
 \begin{tabular}{ c @{\hspace{+2mm}} c @{\hspace{3mm}} c @{\hspace{3mm}} c @{\hspace{2mm}} c}
\\  $a_{1}$ \vspace{1mm}
 \\  $a_{2}$ & 1 
 \\  $a_{3}$ & $\vdots$ & $\ddots$
\\  $\vdots$ & $\vdots$ & &$\ddots$ 
 \\  $a_{n+2}$ & 1 & $\hdots$ & $\hdots$ & 1
\\  &\multicolumn{4}{c}{$\underbrace{\;\;\;\;\;\;\;\;\;\;\;\;\;\;\;\;\;}_\text{n+1}$} 
 \end{tabular}  
\end{equation}
where for $a_3 \ldots a_{n+2}$ we take arbitrary values. If $n=0$, then $(a_1,a_2)=(0,1)$, while if $n>0$, $a_1,a_2$ are determined
by
\begin{eqnarray}
a_1&=&1+ \sum_{j=2}^{n+1}{\binom{n}{j-1}}a_j \,\,\,\mathrm{mod} \,2, \\
a_2&=& \sum_{j=2}^{n+1}{\binom{n}{j-1}}a_{j+1} \,\,\,\mathrm{mod} \,2. 
\end{eqnarray}
\end{prop}
\begin{prop}\label{preim101-2d}
 The set $\g^{-n}(\begin{smallmatrix} 1 &\\0 & 1 \\ \end{smallmatrix})$ consists of all blocks of the following form:
\begin{equation} \label{preim101-2d-block}
 \begin{tabular}{ c @{\hspace{+2mm}} c @{\hspace{3mm}} c @{\hspace{3mm}} c @{\hspace{2mm}} c}
\\  $a_{1}$ \vspace{1mm}
 \\  $a_{2}$ & 1 
 \\  $a_{3}$ & $\vdots$ & $\ddots$
\\  $\vdots$ & $\vdots$ & &$\ddots$ 
 \\  $a_{n+2}$ & 1 & $\hdots$ & $\hdots$ & 1
\\  &\multicolumn{4}{c}{$\underbrace{\;\;\;\;\;\;\;\;\;\;\;\;\;\;\;\;\;}_\text{n+1}$} 
 \end{tabular}  
\end{equation}
where for $a_3 \ldots a_{n+2}$ we take arbitrary values. If $n=0$, then $(a_1,a_2)=(1,0)$, while if $n>0$, $a_1,a_2$ are determined
by
\begin{eqnarray}
a_1&=& \sum_{j=2}^{n+1}{\binom{n}{j-1}}a_j \,\,\,\mathrm{mod} \,2, \\
a_2&=&1+ \sum_{j=2}^{n+1}{\binom{n}{j-1}}a_{j+1} \,\,\,\mathrm{mod} \,2. 
\end{eqnarray}
\end{prop}
Proofs of these propositions can be constructed using very similar argument as in Proposition \ref{preim001-2d}, 
and will thus be omitted.  Note that in each case, there are exactly $n$ arbitrary values in preimages, so there are $2^n$ of them,
and therefore we have
\begin{equation}  
  \text{card}[\g^{-n} \big(\begin{smallmatrix} 0 &\\1 & 1 \\ \end{smallmatrix} \big) ] = \text{card}[\g^{-n} \big(\begin{smallmatrix} 1 &\\0 & 1 \\ \end{smallmatrix} \big) ] = 2^n. 
\end{equation}

We now have enough information to obtain the number of preimages of each of the remaining basic blocks, since, similarly as in one dimensional case, 
block probabilities are interrelated via consistency conditions. 
\begin{align}
P_{n-1}(\begin{smallmatrix} 0 &\\0 & 0 \\ \end{smallmatrix}) + P_{n-1}(\begin{smallmatrix} 0 &\\0 & 1 \\ \end{smallmatrix}) + P_{n-1}(\begin{smallmatrix} 0 &\\1 & 0 \\ \end{smallmatrix}) + P_{n-1}(\begin{smallmatrix} 0 &\\1 & 1 \\ \end{smallmatrix}) &= P_{n}(0), \\
 P_{n-1}(\begin{smallmatrix} 0 &\\0 & 0 \\ \end{smallmatrix}) + P_{n-1}(\begin{smallmatrix} 0 &\\1 & 0 \\ \end{smallmatrix}) + P_{n-1}(\begin{smallmatrix} 1 &\\0 & 0 \\ \end{smallmatrix}) + P_{n-1}(\begin{smallmatrix} 1 &\\1 & 0 \\ \end{smallmatrix}) &= P_{n}(0), \\
P_{n-1}(\begin{smallmatrix} 0 &\\0 & 0 \\ \end{smallmatrix}) + P_{n-1}(\begin{smallmatrix} 0 &\\0 & 1 \\ \end{smallmatrix}) + P_{n-1}(\begin{smallmatrix} 1 &\\0 & 0 \\ \end{smallmatrix}) + P_{n-1}(\begin{smallmatrix} 1 &\\0 & 1 \\ \end{smallmatrix}) &= P_{n}(0).
\end{align}
Results are summarized in Table~2, where numbers of preimages  for all eight basic blocks are shown.

\begin{table}[h]
\centering
\begin{tabular}{|| c | c || c | c ||}
\hline
\(b\) & \(\text{card}[ \g^{-n}( b)]\) & \(b\) & \(\text{card}[ \g^{-n}( b)]\) \\
 \hline
 \(\begin{matrix} 0 &\\0 & 0 \\ \end{matrix}\) & \(2^{\frac{n^2+5n}{2}} \Big(8-3 \sum_{i=0}^{n} 2^{-\frac{i(i+3)}{2}} \Big) -4\cdot2^n\) & \(\begin{matrix} 1 &\\0 & 0 \\ \end{matrix}\) & \( 2^{\frac{n^2+5n}{2}} \Big(-4+4 \sum_{i=0}^{n} 2^{-\frac{i(i+3)}{2}}\Big)+2^n \)\\ \hline
 
 \(\begin{matrix}  0 &\\0 & 1 \\ \end{matrix}\) & \( 2^{\frac{n^2+5n}{2}} \Big(\sum_{i=0}^{n} 2^{-\frac{i(i+3)}{2}} \Big)\) & \(\begin{matrix} 1 &\\0 & 1 \\ \end{matrix}\) & \( 2^n \)\\ \hline

\(\begin{matrix} 0 &\\1 & 0 \\ \end{matrix}\) & \( 2^{\frac{n^2+5n}{2}} \Big(-4+4 \sum_{i=0}^{n} 2^{-\frac{i(i+3)}{2}}\Big) + 2^n \) & \(\begin{matrix} 1 &\\1 & 0 \\ \end{matrix}\) & \( 2^{\frac{n^2+5n}{2}} \Big( 4-3 \sum_{i=0}^{n} 2^{-\frac{i(i+3)}{2}} \Big) \) \\ \hline

 \(\begin{matrix}0 &\\1 & 1 \\ \end{matrix}\) & \( 2^n \) & \(\begin{matrix}1 &\\1 & 1 \\ \end{matrix}\) & \( 2^{\frac{n^2+5n}{2}} \Big( 4-3 \sum_{i=0}^{n} 2^{-\frac{i(i+3)}{2}} \Big) \)  \\ 
\hline \hline
\end{tabular}
\caption[]{Number of preimages of basic blocks for two-dimensional rule 130.}
\end{table}

\section*{Dependence on the initial density in two dimensions}
Obviously, for any triangular block $b$ we have
\begin{equation}
 P_n(b)=\sum_{a \in \g^{-n}(b)} P_0(a).
\end{equation}
Using Proposition \ref{preim001-2d}
one can therefore compute $P_{n}(\begin{smallmatrix} 0 &\\0 & 1 \\ \end{smallmatrix})$ as follows:
\begin{equation}
 P_{n}(\begin{smallmatrix} 0 &\\0 & 1 \\ \end{smallmatrix})=
\sum_{i=0}^{n}
\sum_{a_3\ldots a_{i+2} \in \{0,1\}}  P_{0}\left(\begin{smallmatrix} a_1\\a_2 \end{smallmatrix}\right)
P_{0}\left(\begin{smallmatrix} a_3\\a_4\\ \vdots\\ a_{i+2}\end{smallmatrix}\right)
P_0
\left(  \underbrace{ 
 \begin{smallmatrix}
 1 \\ \vdots & \ddots \\ 1 & \hdots & 1
 \end{smallmatrix}  }_{i+1} \right),
\end{equation}
where $a_1,a_2$ are given by eq. (\ref{a1}) and (\ref{a2}). Similarly, Proposition~\ref{preim111-2d} yields

\begin{equation}
 P_{n}(\begin{smallmatrix} 1 &\\1 & 1 \\ \end{smallmatrix})=
\sum_{i=0}^{n}
\sum_{b_3\ldots b_{i+2} \in \{0,1\}}  P_{0}\left(\begin{smallmatrix} b_1\\b_2 \end{smallmatrix}\right)
P_{0}\left(\begin{smallmatrix} b_3\\b_4\\ \vdots\\ b_{i+2}\end{smallmatrix}\right)
P_0
\left(  \underbrace{ 
 \begin{smallmatrix}
 1 \\ \vdots & \ddots \\ 1 & \hdots & 1
 \end{smallmatrix}  }_{i+1} \right)
- \sum_{i=0}^{n} P_0
\left(  \underbrace{ 
 \begin{smallmatrix}
 1 \\ \vdots & \ddots \\ 1 & \hdots & 1
 \end{smallmatrix}  }_{i+2} \right)
+  P_0
\left(  \underbrace{ 
 \begin{smallmatrix}
 1 \\ \vdots & \ddots \\ 1 & \hdots & 1
 \end{smallmatrix}  }_{n+2} \right),
\end{equation}
where $b_1,b_2$ are given by eq. (\ref{b1}) and (\ref{b2}). Now, since $P_n(1)=
P_{n-1}(\begin{smallmatrix} 1 &\\1 & 1 \\ \end{smallmatrix}) +P_{n-1}(\begin{smallmatrix} 1 &\\0 & 0 \\ \end{smallmatrix})$,
and 
\begin{equation}
  P_0
\left(  \underbrace{ 
 \begin{smallmatrix}
 1 \\ \vdots & \ddots \\ 1 & \hdots & 1
 \end{smallmatrix}  }_{n} \right) =\rho^{n(n+1)/2},
\end{equation}
we obtain
\begin{align} \nonumber
 P_{n}(1)&=
\sum_{i=0}^{n-1}
\sum_{a_3\ldots a_{i+2} \in \{0,1\}}  P_{0}\left(\begin{smallmatrix} a_1\\a_2 \end{smallmatrix}\right)
P_{0}\left(\begin{smallmatrix} a_3\\a_4\\ \vdots\\ a_{i+2}\end{smallmatrix}\right)
\rho^{(i+1)(i+2)/2}
\\
&+\sum_{i=0}^{n-1}
\sum_{b_3\ldots b_{i+2} \in \{0,1\}}  P_{0}\left(\begin{smallmatrix} b_1\\b_2 \end{smallmatrix}\right)
P_{0}\left(\begin{smallmatrix} b_3\\b_4\\ \vdots\\ b_{i+2}\end{smallmatrix}\right)
\rho^{(i+1)(i+2)/2}
- \sum_{i=0}^{n-1} \rho^{(i+2)(i+3)/2}
+\rho^{(n+1)(n+2)/2},  \label{densityresponse2d}
\end{align}
where, again, $a_1,a_2,b_1$, and $b_2$ are determined by eq. (\ref{a1}), (\ref{a2}), (\ref{b1}) and (\ref{b2}). The above
is an exact equation of the response curve, although, unfortunately, it is not possible to calculate the double sums in a closed form. Nevertheless, for a given $n$, we can calculate and plot $c_n$ versus $\rho$, providing that $n$ is not too large. This has been done for $n=12$, as shown in Figure~\ref{respurvefig2d}. Again, there is excellent agreement between the theoretical curve representing infinite lattice and computer simulations done on a finite lattice. 
\begin{figure}
\begin{center}
\includegraphics[scale=0.4]{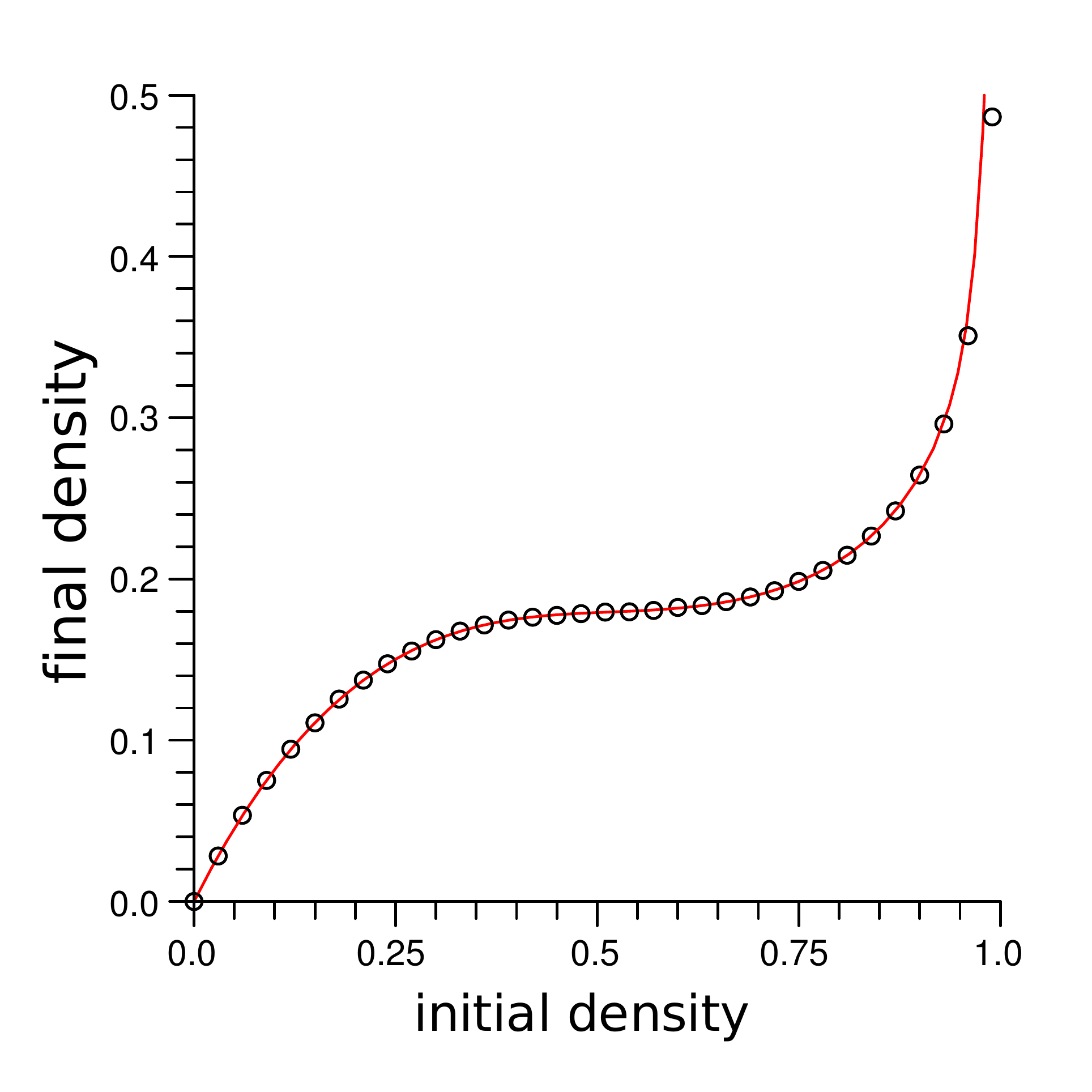}	
\end{center}
\caption{Plots of the response curve, i.e., dependence of $c_n$ on $\rho$ for two-dimensional rule 130. Circles  correspond to computer simulations, while the continuous line represents theoretical curve for $n=12$.}\label{respurvefig2d}
\end{figure}

It is also possible to compute $P_n(1)$ for a special case of $\rho=1/2$. In this case $P_{0}\left(\begin{smallmatrix} a_1\\a_2 \end{smallmatrix}\right)
=1/4$ for any $a_1, a_2$, hence
\begin{equation}
 P_n(1)=2\sum_{i=0}^{n-1} \frac{1}{4} 2^{-(i+1)(i+2)/2} - \sum_{i=0}^{n-1} 2^{-(i+2)(i+3)/2}
+2^{-(n+1)(n+2)/2},
\end{equation}
which simplifies to
\begin{equation} \label{dens2d}
P_n(1)= \frac{1}{2}-\frac{1}{4}\sum_{i=0}^{n-1} 2^{-\frac{i(i+3)}{2}}.
\end{equation}
The asymptotic density is then given by
\begin{equation}
c_\infty = \lim_{n \to \infty} P_n(1) = \frac{1}{2}-\frac{1}{4}\sum_{i=0}^{\infty} 2^{-\frac{i(i+3)}{2}},
\end{equation}
where, again, it is not possible to compute the sum in a closed form. We can, however,  approximate the infinite sum by taking the finite
number of terms and use the integral bounds for the remainder,
\begin{equation} \int_{k+1}^{\infty} 2^{-\frac{x(x+3)}{2}} \; dx \leq  \sum_{i=k}^{\infty} 2^{-\frac{i(i+3)}{2}} \leq \int_{k}^{\infty} 2^{-\frac{x(x+3)}{2}} \; dx. \end{equation}
The above integrals can be expressed in terms of the error function, hence for any positive integer $k$ we obtain
\begin{equation} 
\sum_{i=0}^{k-1} 2^{-\frac{i(i+3)}{2}}+
\frac{\sqrt{\pi} 2^{1/8} }{a} \big(1-\text{erf}(ak+\tfrac{5}{2} a)\big) \leq  
\sum_{i=0}^{\infty} 2^{-\frac{i(i+3)}{2}} \leq 
\sum_{i=0}^{k-1} 2^{-\frac{i(i+3)}{2}}+
\frac{\sqrt{\pi} 2^{1/8} }{a} \big(1-\text{erf}(ak+\tfrac{3}{2} a)\big), \end{equation}
where $\text{erf}(x)=\frac{2}{\sqrt{\pi}} \int_{0}^{x} e^{-t^2} dt$, and $a=\tfrac{\sqrt{2 \ln{2} }}{2}$.
The sum converges very fast, and, for example, using the above inequalities for  $k=5$ we obtain 
\begin{equation}
 0.1791839087\ldots \leq  c_\infty \leq 0.1791839597\ldots ,\end{equation}
that is, $c_\infty$ with accuracy of seven digits after the decimal point.

\section*{Conclusions and future work}
We demonstrated that response curves are calculable for simple rules such as rule 130. We obtained
exact formulae for response curves for this rule in one (eq. \ref{densityesponse1d}) and two dimensions
(eq.  \ref{densityresponse2d}). Techniques presented here are applicable to a fairly large class of CA rules -- in
case of elementary rules, we conjecture that in about 70\%  of cases density response could be calculated rigorously.
This conjecture is based on numerical experiments and computerized searches for patterns in preimage strings.
These patterns and regularities are often obvious, in other cases they can be discovered by 
constructing minimal finite state machines describing preimage sets or by using heuristic methods.
Once the pattern is discovered, one needs to construct a formal proof that the pattern exists
in all levels of preimages, and then perform calculations of densities. This task must obviously be
tailored to each particular rule, and we currently do not have any general method for doing this.
It seems, however, that rules which belong to Wolfram class I or class  II are more likely to possess 
recognizable patterns in their preimage trees, and thus their density curves are often
computable. 

Response curves clearly deserve further study, and it would be worthwhile to systematically study them for
a large number of CA rules, especially in the context of applications of CA to solving computational problems.
Work in this direction is ongoing, and will be reported elsewhere.

\section*{Acknowledgments}
One of the authors (HF) acknowledges financial support from the Natural
Sciences and Engineering Research Council of Canada (NSERC) in the
form of Discovery Grant. We also wish to thank Shared Hierarchical Academic Research Computing Network
(SHARCNET) for granting us access to high-performance computing facilities and for technical support.
Finally, we would like to thank to anonymous referees for insightful comments which helped to
improve this paper significantly.


\begin{thebibliography}{}

\bibitem[\protect\citeauthoryear{%
Dynkin%
}{%
Dynkin%
}{%
{\protect\APACyear{1969}}%
}]{%
Dynkin69}%
\APACinsertmetastar{%
Dynkin69}%
Dynkin, E.~B.%
%
\newblock{}\BBOP{}1969\BBCP{}.
\newblock{}\Bem{Markov processes-theorems and problems}.
\newblock{}New York: Plenum Press.

\bibitem[\protect\citeauthoryear{%
Fuk{\'s}%
}{%
Fuk{\'s}%
}{%
{\protect\APACyear{2006}}%
}]{%
paper27}%
\APACinsertmetastar{%
paper27}%
Fuk{\'s}, H.%
%
\newblock{}\BBOP{}2006\BBCP{}.
\newblock{}\BBOQ{}Dynamics of the cellular automaton rule 142.\BBCQ{}
\newblock{}\Bem{Complex Systems}, \Bem{16}, 123--138.

\bibitem[\protect\citeauthoryear{%
Fuk{\'s}%
}{%
Fuk{\'s}%
}{%
{\protect\APACyear{2010}}%
}]{%
paper39}%
\APACinsertmetastar{%
paper39}%
Fuk{\'s}, H.%
%
\newblock{}\BBOP{}2010\BBCP{}.
\newblock{}\BBOQ{}Probabilistic initial value problem for cellular automaton
  rule 172.\BBCQ{}
\newblock{}\Bem{DMTCS proc.}, \Bem{AL}, 31-44. (to appear)

\bibitem[\protect\citeauthoryear{%
Fuk{\'s}%
\ \BBA{} Haroutunian%
}{%
Fuk{\'s}%
\ \BBA{} Haroutunian%
}{%
{\protect\APACyear{2009}}%
}]{%
paper34}%
\APACinsertmetastar{%
paper34}%
Fuk{\'s}, H.%
\BCBT{}\ \BBA{} Haroutunian, J.%
%
\newblock{}\BBOP{}2009\BBCP{}.
\newblock{}\BBOQ{}Catalan numbers and power laws in cellular automaton rule
  14.\BBCQ{}
\newblock{}\Bem{Journal of cellular automata}, \Bem{4}, 99-110.

\bibitem[\protect\citeauthoryear{%
Land%
\ \BBA{} Belew%
}{%
Land%
\ \BBA{} Belew%
}{%
{\protect\APACyear{1995}}%
}]{%
LB95}%
\APACinsertmetastar{%
LB95}%
Land, M.%
\BCBT{}\ \BBA{} Belew, R.~K.%
%
\newblock{}\BBOP{}1995\BBCP{}.
\newblock{}\BBOQ{}No perfect two-state cellular automata for density
  classification exists.\BBCQ{}
\newblock{}\Bem{Phys. Rev. Lett.}, \Bem{74}(25), 5148--5150.

\bibitem[\protect\citeauthoryear{%
Wolfram%
}{%
Wolfram%
}{%
{\protect\APACyear{1994}}%
}]{%
Wolfram94}%
\APACinsertmetastar{%
Wolfram94}%
Wolfram, S.%
%
\newblock{}\BBOP{}1994\BBCP{}.
\newblock{}\Bem{Cellular automata and complexity: Collected papers}.
\newblock{}Reading, Mass.: Addison-Wesley.

\end{thebibliography}
\end{document}